\documentclass{llncs}
\usepackage[T1]{fontenc}
\usepackage[utf8]{inputenc}
\usepackage{cite,microtype,graphicx}
\usepackage{hyperref}
\usepackage{amsmath,amssymb}
\usepackage[capitalize,nameinlink]{cleveref}

\spnewtheorem{observation}[lemma]{Observation}{\bfseries}{\itshape}
\title{Lower Bounds for Non-Adaptive\\ Shortest Path Relaxation}

\author{David Eppstein}
\authorrunning{D. Eppstein}
\institute{Computer Science Department\\
University of California, Irvine\\
Irvine, CA 92697, USA\\
\email{eppstein@uci.edu}}

\date{ }

\begin{document}
\maketitle  

\begin{abstract}
We consider single-source shortest path algorithms that perform a sequence of relaxation steps whose ordering depends only on the input graph structure and not on its weights or the results of prior steps. Each step examines one edge of the graph, and replaces the tentative distance to the endpoint of the edge by its minimum with the tentative distance to the start of the edge, plus the edge length. As we prove, among such algorithms, the Bellman–Ford algorithm has optimal complexity for dense graphs and near-optimal complexity for sparse graphs, as a function of the number of edges and vertices in the given graph. Our analysis holds both for deterministic algorithms and for randomized algorithms that find shortest path distances with high probability.
\end{abstract}

\section{Introduction}

Dijkstra's algorithm finds shortest paths in directed graphs when all edge weights are non-negative, but the problem becomes more difficult when negative edge weights (but not negative cycles) are allowed. In this case, despite recent breakthroughs on near-linear time bounds for graphs with small integer edge weights~\cite{BerNanWul-FOCS-22}, the best strongly-polynomial time bound for single-source shortest paths remains that of the Bellman–Ford algorithm~\cite{Bel-QAM-58,ForFul-62,Moo-ISST-57}, which takes time $O(mn)$ on graphs with $m$ edges and $n$ vertices, or $O(n^3)$ on dense graphs.

Both Dijkstra's algorithm and the Bellman–Ford algorithm (as well as an unnamed linear-time algorithm for single-source shortest paths in directed acyclic graphs) can be unified under the framework of \emph{relaxation algorithms}, also called \emph{label-correcting algorithms}~\cite{DeoPan-Nw-84}. These algorithms initialize tentative distances $D[v]$ from the source vertex to each other vertex $v$, by setting $D[s]=0$ and $D[v]=+\infty$ for $v\ne s$. Then, they repeatedly \emph{relax} the edges of the graph. This means, that for a given edge $u\to v$, the algorithm updates $D[v]$ to $D[u]+\operatorname{length}(u\to v)$. In Dijkstra's algorithm, each edge $u\to v$ is relaxed once, in sorted order by the tentative distance $D[u]$. In the Bellman–Ford algorithm, an edge can be relaxed many times. The algorithm starts with the tentative distance equal to the correct distance for $s$, but not for the other vertices. Whenever the algorithm relaxes an edge $u\to v$ in the shortest path tree, at a time when $u$ already has the correct distance, the tentative distance to $v$ becomes correct as well. Thus, the goal in designing the algorithm is to perform these distance-correcting relaxations while wasting as little effort as possible on other relaxations that do not correct any distance, and on the overhead in selecting which relaxation to perform.

We would like to prove or disprove the optimality of the Bellman--Ford algorithm among a general class of strongly-polynomial shortest path algorithms, without restricting the types of computation such an algorithm can perform, but such a result appears to remain far out of reach.
Instead, in this work we focus only on relaxation algorithms, asking: how few relaxation steps are needed? Note that, without further assumptions, a shortest path algorithm could ``cheat'', computing a shortest path tree in some other way and then performing only $n-1$ relaxation steps in a top-down traversal of a shortest path tree. To focus purely on relaxation, and prevent such cheating, we consider \emph{non-adaptive relaxation algorithms}, in which the sequence of relaxation steps is determined only by the structure of the given graph, and not on its weights nor on the outcome of earlier relaxation steps. Dijkstra's algorithm is adaptive, but the linear-time DAG algorithm is non-adaptive. Another example of  a non-adaptive algorithm comes from past work on the graphs in which, like DAGs, it is possible to relax every edge once in a fixed order and guarantee that all tentative distances are correct~\cite{HadSch-SIDMA-88}. As usually described, the Bellman–Ford algorithm is adaptive. Its typical optimizations include adaptive rules that disallow repeatedly relaxing any edge $u\to v$ unless the tentative distance to $u$ has decreased since the previous relaxation, and that stop the entire algorithm when no more allowed relaxations can be found. However, its same asymptotic time bounds can be achieved by a non-adaptive version of the Bellman–Ford algorithm, with a \emph{round-robin} relaxation sequence, one that merely repeats $n-1$ rounds of relaxing all edges in the same order per round. A non-adaptive asynchronous distributed form of the Bellman–Ford algorithm is widely used in \emph{distance vector routing} of internet traffic, to maintain paths of minimum hop count between major internet gateways~\cite{RFC-1058}.

\subsection{Known Upper Bounds}

We do not require non-adaptive relaxation algorithms to be round-robin, but we are unaware of any way to take advantage of this extra flexibility. Nevertheless, among round-robin algorithms, there is still freedom to choose the ordering of edges within each round, and this freedom can lead to improved constant factors in the number of relaxation steps performed by the Bellman–Ford algorithm.

Yen~\cite{Yen-75} described a method based on the following idea. Choose an arbitrary linear ordering for the vertices, and partition the edges into two subsets: the edges that are directed from an earlier vertex to a later vertex in the ordering, and the edges that are directed from a later vertex to an earlier vertex. Both of these two edge subsets define directed acyclic subgraphs of the given graph, with the chosen linear ordering or its reverse as a topological ordering. Use a round-robin edge ordering that first relaxes all of the edges of the first subgraph, in its topological order, and then relaxes all of the edges of the second subgraph, in its topological order. If any shortest path is divided into contiguous subpaths that lie within one of these two DAGs, then each two consecutive subpaths from the first and second DAG will be relaxed in order by each round of the algorithm. In the worst case, there is a single shortest path of $n-1$ edges, alternating between the two DAGs, requiring $\lceil n/2\rceil$ rounds of relaxation. For complete directed graphs, this method uses $\bigl(\tfrac12+o(1)\bigr)n^3$ relaxation steps, instead of the $\bigl(1+o(1)\bigr)n^3$ that might be used by a less-careful round-robin method.

As we showed in earlier work~\cite{BanEpp-ANALCO-12}, an additional constant factor savings can be obtained by a randomized algorithm that selects from a random distribution of non-adaptive relaxation sequences, and that obtains a correct output with high probability rather than with certainty. To do so, use Yen's method, but choose the vertex ordering as a uniformly random permutation of the vertices, rather than arbitrarily. In any shortest path tree, each vertex with more than one child reduces the number of steps from the source to the deepest leaf by one, reducing the number of alternations between the two DAGs. For each remaining vertex with one child in the tree, the probability that it lies between its parent and child in the randomly selected ordering is $\tfrac13$, and when this happens, it does not contribute to the bound on the number of alternations. With high probability, the number of these non-contributing vertices is close to one third of the single-child vertices. Therefore, with high probability, the maximum number of alternations between the two DAGs among paths on the shortest path tree is $\bigl(\tfrac23+o(1)\bigr)n$, and an algorithm that uses this method to perform $\bigl(\tfrac13+o(1)\bigr)n^3$ relaxation steps will find the correct shortest paths with high probability.

The worst-case asymptotic time of these methods remains $O(n^3)$ for complete graphs, and $O(mn)$ for arbitrary graphs with $m$ vertices and $n$ edges. Both Yen's method and the randomized permutation method can also be used in adaptive versions of the Bellman–Ford algorithm, with better constant factors and in the randomized case leading to a Las Vegas algorithm rather than a Monte Carlo algorithm, but it is their non-adaptive variants that concern us here.

\subsection{New Lower Bounds}

We provide the following results:
\begin{itemize}
\item Any deterministic non-adaptive relaxation algorithm for single-source shortest paths on a complete directed graph with $n$ vertices must use $\bigl(\tfrac16-o(1)\bigr)n^3$ relaxation steps.
\item Any randomized non-adaptive relaxation algorithm for shortest paths on a complete directed graph with $n$ vertices, that with high probability sets all distances correctly, must use $\bigl(\tfrac1{12}-o(1)\bigr)n^3$ relaxation steps.
\item For any $m$ and $n$ with $n\le m\le2\tbinom{n}{2}$, there exists a directed graph on $m$ edges and $n$ vertices on which any deterministic or high-probability randomized non-adaptive relaxation algorithm for shortest paths must use $\Omega(mn/\log n)$ relaxation steps. When $m=\Omega(n^{1+\varepsilon})$ for some $\varepsilon>0$, the lower bound improves to $\Omega(mn)$.
\end{itemize}

These lower bounds hold even on graphs for which all edges weights are zero and one, for which an adaptive algorithm, Dial's algorithm, can find shortest paths in linear time~\cite{Dia-CACM-69}.

\subsection{Related Work}

Although we are not aware of prior work in the precise model of computation that we use, variants of the Bellman–Ford algorithm have been studied and shown optimal for some other related problems:
\begin{itemize}
\item The $k$-walk problem asks for a sequence of exactly $k$ edges, starting and one vertex and ending at the other, allowing repeated edges. The Bellman–Ford algorithm can be modified to find the shortest $k$-walk between two vertices in time $O(kn^2)$, non-adaptively. In any non-adaptive relaxation algorithm, the only arithmetic operations on path lengths and edge weights are addition and minimization, and these operations are performed in a fixed order. Therefore, the sequence of these operations can be expanded into a circuit, with two kinds of gates: minimization and addition. The resulting $(\min,+)$-circuit model of computation is somewhat more general than the class of relaxation algorithms, because the sequence of operations performed in this model does not need to come from a sequence of relaxation steps. The $k$-walk version of the Bellman–Ford algorithm is nearly optimal in the $(\min,+)$-circuit model: circuit size $\Omega\bigl(k(n-k)n\bigr)$ is necessary~\cite{JukSch-TCS-16}.  However, this $k$-walk problem is different from the shortest path problem, so this bound does not directly apply to shortest paths.
\item Under conditional hypotheses that are standard in fine-grained complexity analysis, the $O(km)$ time of Bellman–Ford for finding paths of at most $k$ steps, for graphs of $m$ edges, is again nearly optimal: neither the exponent of $k$ nor the exponent of $m$ can be reduced to a constant less than one. For large-enough $k$, the shortest path of at most $k$ steps is just the usual shortest path, but this lower bound applies only for choices of $k$ that are small enough to allow the result to differ from the shortest path~\cite{Pol-22}. 
\item Another related problem is the all hops shortest path problem, which asks to simultaneously compute $k$ paths, having distinct numbers of edges from one to a given parameter $k$. Again, this can be done in time $O(km)$ by a variant of the Bellman–Ford algorithm, and it has an unconditional $\Omega(km)$ lower bound for algorithms that access the edge weights only by path length comparisons, as Bellman–Ford does~\cite{GueOrd-IATW-02,CheAns-ICL-04}. Because it demands multiple paths as output, this lower bound does not apply to algorithms that compute only a single shortest path. 
\item Meyer et al.~\cite{MeyNegWei-TAPAS-11} study a version of the Bellman–Ford algorithm, in which edges are relaxed in a specific (adaptive) order. They construct sparse graphs, with $O(n)$ edges, on which this algorithm takes $\Omega(n^2)$ time, even in the average case for edge weights uniformly drawn from a unit interval. This bound applies only to this algorithm and not to other relaxation orders.
\end{itemize}

\section{Deterministic Lower Bound for Complete Graphs}

The simplest of our results, and the prototype for our other results, is a lower bound on the number of relaxations needed by a deterministic non-adaptive relaxation algorithm, in the worst case, on a complete directed graph with $n$ vertices.

\begin{theorem}
\label{thm:complete-deterministic}
Any deterministic non-adaptive relaxation algorithm for single-source shortest paths on a complete directed graph with $n$ vertices must use at least $\bigl(\tfrac16-o(1)\bigr)n^3$ relaxation steps.
\end{theorem}

\begin{proof}
Fix the sequence $\sigma$ of relaxation steps chosen by any such algorithm. We will find an assignment of weights for the complete directed graph, such that the distances obtained by the relaxation algorithm are not all correct until $\bigl(\tfrac16-o(1)\bigr)n^3$ relaxation steps have taken place. Therefore, in order for the algorithm to be correct, it must make this many steps. For the weights we choose, the shortest path tree will form a single directed path, of $n-1$ edges, starting at the source vertex. In order for the relaxation algorithm to achieve correct distances to all vertices, its sequence of relaxations must include a subsequence consisting of all path edges in order. The weights of these edges are unimportant (because we are considering only non-adaptive algorithms) so we may set all path edges to have weight zero and all other edges to have weight one.

To determine this path, we choose one at a time its edges in even positions: its second, fourth, sixth, etc., edge. These chosen edges include every vertex in the path, so choosing them will also determine the edges in odd positions. When choosing the $i$th edge (for an even number $i$), we make the choice greedily, to maximize the position in $\sigma$ of the step that relaxes this edge and makes its endpoint have the correct distance. Let $s_i$ denote this position, with $s_0=0$ as a base case recording the fact that, before we have relaxed any edges, the source vertex already has the correct distance. Then the length of $\sigma$ is at least equal to the telescoping sum
\[(s_2-s_0) + (s_4-s_2) + (s_6-s_4)+\cdots.\]

When choosing edge $i$, for an even position $i$, there are $i-1$ earlier vertices, whose position in the shortest path is already determined, and $n-i+1$ remaining vertices. Between step $s_{i-2}$ and step $s_i$ of the relaxation sequence $\sigma$, it must relax all $n-i+1$ edges from the last endpoint of edge $i-2$ to one of these remaining vertices, and all $2\tbinom{n-i+1}{2}$ edges between pairs of the vertices that remain to be corrected. For, if it did not do so, there would be an edge that it had not relaxed, and choosing this edge next would cause $s_i$ to be greater; but this would violate the greedy choice of edge $i$ to make $s_i$ as large as possible. Therefore,
\[s_{i}-s_{i-2}\ge (n-i+1)+2\binom{n-i+1}{2} = (n-i+1)^2.\]
 
Summing over all $\lfloor(n-1)/2\rfloor$ choices of edges in even positions gives, as a lower bound on the total number of relaxation steps,
 \[
 \sum_{i=2,4,6,\dots} s_i-s_{i-2}\ge \sum_{i=2,4,6,\dots}(n-i+1)^2= \frac{n^3-n}{6},
 \]
where the closed form for the summation follows easily by induction.
\end{proof}

\section{Randomized Lower Bound for Complete Graphs}

It does not make much sense to consider expected time analysis for non-adaptive algorithms, because these algorithms have a fixed stopping time (determined as a function of the given graph), and we want their output to be correct with high probability rather than in any expected sense. Nevertheless, it is often easier to lower-bound the expected behavior of randomized algorithms, by using Yao's principle~\cite{Yao-FOCS-77}, according to which the expected cost of a randomized algorithm on its worst-case input can be lower bounded by the cost of the best deterministic algorithm against any random distribution of inputs.

In order to convert high-probability time bounds into expectations, we consider randomized non-adaptive algorithms that are guaranteed to produce the correct distances, and we define the \emph{reduced cost} of such an algorithm to be the number of relaxations that it performs until all distances are correct, ignoring any remaining relaxations after that point.

\begin{lemma}
If a randomized non-adaptive relaxation algorithm $\mathcal{A}$ takes $s(G)$ steps on any weighted input graph $G$ and computes all distances from the source vertex correctly with probability $1-o(1)$, then there exists a randomized non-adaptive relaxation algorithm $\mathcal{B}$ that is guaranteed to produce correct distances and whose expected reduced cost, on weighted graphs $G$ with $n$ vertices and $m$ edges, is at most $s(G)+o(mn)$.
\end{lemma}

\begin{proof}
Construct algorithm $\mathcal{B}$ by using the relaxation sequence from algorithm $\mathcal{A}$, appending onto it the sequence of relaxations from a conventional non-adaptive deterministic Bellman–Ford algorithm. Then with probability $1-o(1)$ the relaxed cost of $\mathcal{B}$ counts only the relaxation sequence from algorithm $\mathcal{A}$, of length $s(G)$. With probability $o(1)$ the relaxed cost extends into the deterministic Bellman–Ford part of the sequence, of length $O(mn)$. Because this happens with low probability, its contribution to the expected reduced cost is $o(mn)$.
\end{proof}

\begin{corollary}
\label{cor:red}
Any lower bound on expected reduced cost is also a valid lower bound, up to an additive $o(mn)$ term, on the number of relaxation steps for a randomized non-adaptive relaxation algorithm that produces correct distances with high probability.
\end{corollary}

With this conversion to expected values in hand, we may now formulate Yao's principle as it applies to our problem. We need the following notation:

\begin{definition}
For any graph $G$, with a specified source vertex, let $W_G$ be the family of assignments of real weights to edges of~$G$. Let $\mathcal{D}_G$ be the family of probability distributions of weights in $W_G$, and let $\Sigma_G$ be the class of relaxation sequences on $G$ that are guaranteed to produce correct distances from the specified source vertex. For any randomized non-adaptive relaxation algorithm $\mathcal{A}$ and weight vector $w\in W_G$, let $r_G(\mathcal{A},w)$ denote the expected reduced cost of running algorithm $\mathcal{A}$ on $G$ with edges weighted by $w$. For $\sigma\in\Sigma_G$ and $D\in\mathcal{D}_G$ let $\rho_G(\sigma,D)$ be the expected reduced cost of sequence $\sigma$ on weight vectors drawn from $D$. 
\end{definition}

\begin{lemma}[Yao's principle]
\label{lem:yao}
For any graph $G$ with specified source vertex, and any randomized non-adaptive relaxation algorithm $\mathcal{A}$,
\[
\min_\mathcal{A}\max_{w\in W_G} r_G(\mathcal{A},w) = \max_{D\in\mathcal{D}_G}\min_{\sigma\in\Sigma_G}\rho_G(\sigma,D).
\]
\end{lemma}

\begin{proof}
This is just the minimax principle for zero-sum games, applied to a game in which one player chooses a relaxation sequence $\sigma\in\Sigma_G$, the other player chooses a weight vector $w\in W_G$, and the outcome of the game is the reduced cost for $\sigma$ on $w$. According to that principle, the value of the best mixed strategy for the sequence player, against its worst-case pure strategy (the left hand side of the equality in the lemma) equals the value of the best mixed strategy for the weight player, against its worst-case pure strategy (the right hand side).
\end{proof}

\begin{corollary}
\label{cor:yao}
For any weight distribution $D\in\mathcal{D}_G$, $\min_{\sigma\in\Sigma_G}\rho_G(\sigma,D)$ is a valid lower bound on the expected reduced cost of any randomized non-adaptive relaxation algorithm that is guaranteed to produce correct distances.
\end{corollary}

\begin{proof}
An arbitrary algorithm $\mathcal{A}$ can only have a greater or equal value to the left hand side of \cref{lem:yao}, and an arbitrary weight distribution $D$ can only have a smaller or equal value to the right hand side. So the expected reduced cost of the algorithm, on a worst-case input, can only be greater than or equal to the value given for $D$ in the statement of the corollary.
\end{proof}

\begin{theorem}
\label{thm:complete-random}
Any randomized non-adaptive relaxation algorithm for shortest paths on a complete directed graph with $n$ vertices, that with high probability sets all distances correctly, must use at least $\bigl(\tfrac1{12}-o(1)\bigr)n^3$ relaxation steps.
\end{theorem}

\begin{proof}
We apply \cref{cor:yao} to a weight distribution $D$ defined as follows: we choose a random permutation of the vertices of the given complete graph, starting with the source vertex, we make the weight of edges connecting consecutive vertices in order along this permutation zero, and we make all other weights one. Thus, each weighting of the complete graph drawn from this distribution will have a unique shortest path tree in the form of a single path, with all paths from the source vertex equally likely. For any weight vector $w$ drawn from $D$, let $\pi_w$ be this path.

Let $\sigma$ be any relaxation sequence in $\Sigma_D$. As in the proof of \cref{thm:complete-deterministic}, we define $s_i$ (for a weight vector $w$ to be determined) to be the step at which the second endpoint of the $i$th edge of $\pi_w$ has its shortest path distance set correctly.

Let $C_i$ denote the conditional probability distribution obtained from $D$ by fixing the choice of the first $i$ edges of $\pi_w$. Under condition $C_i$, the remaining $n-i-1$ vertices remain equally likely to be permuted in any order. There are $2\tbinom{n-i-1}{2}$ choices for edge $i+2$, each of which is equally likely. Therefore, the expected value of $s_{i+2}-s_i$ is greater than or equal to the average, among these edges, of their distance along sequence~$\sigma$ from position $s_i$. (It is greater than or equal, rather than equal, because this analysis does not take into account the requirement that edge $i+1$ must be relaxed first, before we relax edge $i+2$.) Sequence $\sigma$ can minimize this average if, in $\sigma$, the next $2\tbinom{n-i-1}{2}$ relaxation steps after $s_i$ are exactly these distinct edges. When $\sigma$ packs the edges in this minimizing way, the average is $2\tbinom{n-i-1}{2}/2$; for other sequences it can only be greater. Therefore,
\[
E[s_{i+2}-s_i\mid C_i]\ge \binom{n-i-1}{2}.
\]
Summing these expected differences, over the sequence of values $s_i$ for even $i$, and applying \cref{cor:red} and \cref{cor:yao}, gives the result.
\end{proof}

\section{Lower Bounds for Incomplete Graphs}

In our lower bounds for complete graphs, the edges in even and odd positions of the shortest paths perform very different functions. The edges in even positions are the ones that, at each step in the shortest path, force the relaxation sequence to have a large subsequence of relaxation steps. Intuitively, this is because there are many possible choices for the edge at the next step and all of these possibilities (in the deterministic bound) or many of these possibilities (in the randomized bound) must be relaxed before reaching the edge that is actually chosen. The edges in odd positions, on the other hand, do not contribute much directly to the length of the sequence of relaxation steps. Instead, they are used to connect the edges in the even positions into a single shortest path.

\begin{figure}[t]
\centering\includegraphics[scale=0.166]{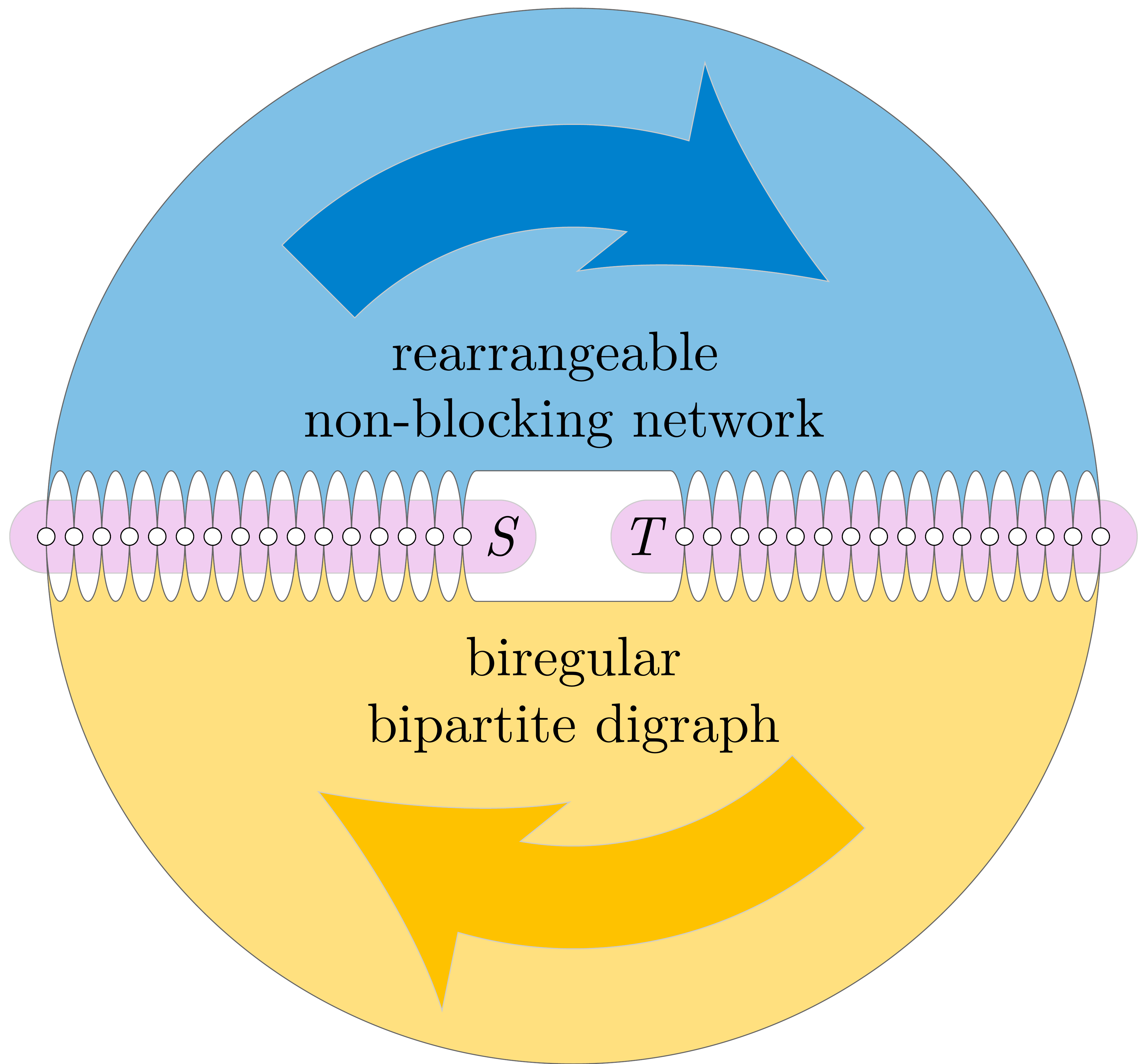}
\caption{Schematic view of the graphs used for our lower bound construction}
\label{fig:schematic}
\end{figure}

To construct graphs that are not complete, for which we can prove analogous lower bounds, we make this dichotomy more explicit. For a chosen ``capacity'' parameter $c$, we will construct graphs that have two designated subsets of $c$ vertices, $S$ and $T$ (with the source vertex contained in subset $S$). We will connect the vertices in $T$ to the vertices in $S$ by a biregular bipartite directed graph of some degree $d\approx m/2c$, a graph in which each vertex in $T$ has exactly $d$ outgoing neighbors and each vertex in $S$ has exactly $d$ incoming neighbors. This biregular graph will perform the function of the even position edges in our complete graph lower bounds: it will have many edges to choose from, forcing any relaxation algorithm to make a long subsequence of relaxations between each two chosen edges. The detailed structure of this graph is not important for our bounds. In the other direction, from $S$ to $T$, we will construct a special graph with the property that, no matter which sequence of disjoint edges we choose from the biregular graph, we can complete this sequence to a path. A schematic view of this construction is depicted in \cref{fig:schematic}. We begin the more detailed description of this structure by defining the graphs we need to connect from $S$ to $T$. The following definition is standard:

\begin{definition}
A \emph{rearrangeable non-blocking network} of capacity $c$ is a directed graph $G$ with $c$ vertices labeled as inputs, and another $c$ vertices labeled as outputs, with the following property. For all systems of pairs of inputs and outputs that include each input and output vertex at most once, there exists in $G$ a system of vertex-disjoint paths  from the input to the output of each pair.
\end{definition}

\begin{observation}
A complete bipartite graph $K_{c,c}$, with its edges directed from $c$ input vertices to $c$ output vertices, is a rearrangeable non-blocking network of capacity $c$, with $2c$ vertices and $c^2$ edges. In this case, the disjoint paths realizing any system of disjoint input-output pairs is just a matching, formed by the edges from the input to the output in each pair.
\end{observation}

\begin{lemma}
\label{lem:expander}
For any capacity $c$, there exist rearrangeable non-blocking network of capacity $c$ with $O(c\log c)$ vertices and edges.
\end{lemma}

Pippenger~\cite{Pip-JCSS-78} credits the proof of \cref{lem:expander} to Beizer~\cite{Bei-SMTA-62}, who used a recursive construction. A more recent construction of Alon and Capalbo~\cite{AloCap-FOCS-07} is based on blowing up an expander graph, producing enough copies of each vertex that a system of edge-disjoint paths in the expander can be transformed into a system of vertex-disjoint paths in the non-blocking network. Their networks are non-blocking in a stronger sense (the vertex-disjoint paths can be found incrementally and efficiently), but we do not need that additional property. A simple counting argument shows that $o(c\log c)$ edges is not possible: to have enough subsets of edges to connect $c!$ possible systems of pairs, the number of edges must be at least $\log_2 c!$.
For non-blocking networks with fewer vertices and more edges we turn to an older construction of Clos~\cite{Clo-BSTJ-53}:

\begin{figure}[t]
\includegraphics[width=\textwidth]{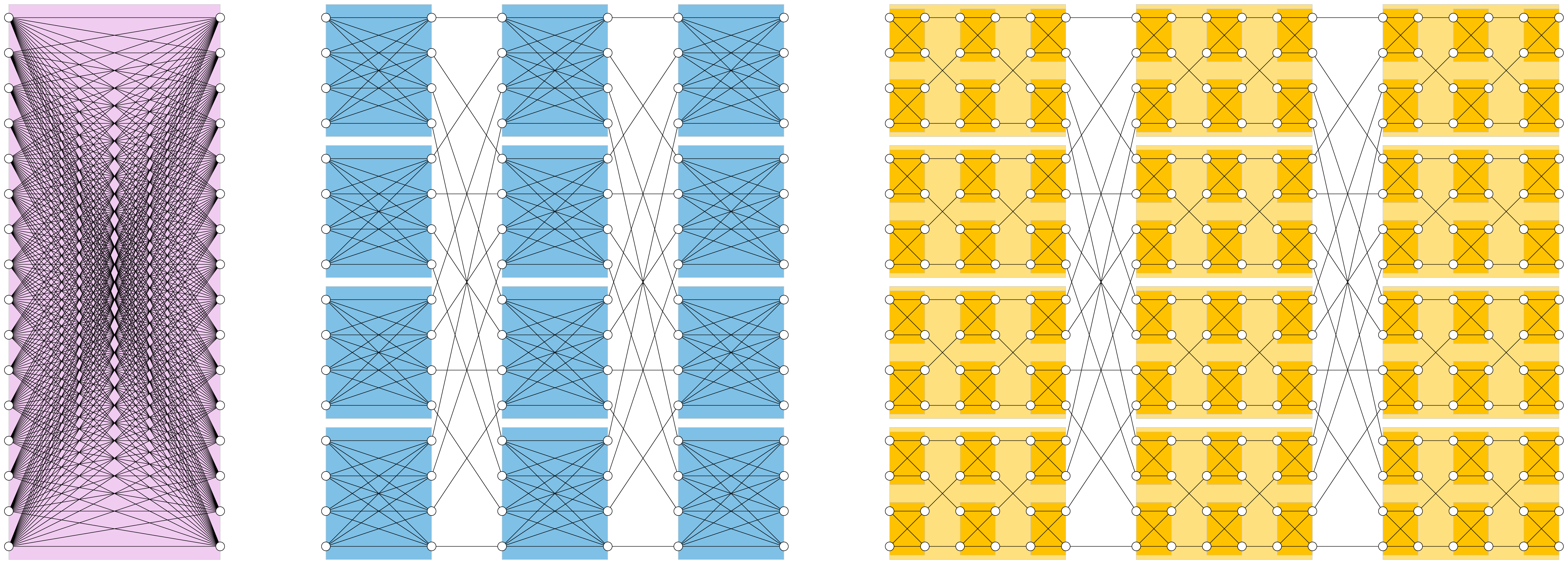}
\caption{Three rearrangeable non-blocking networks of capacity 16.  Each network's input vertices are in its left column and its output vertices are in the right column. Left: Complete bipartite graph. Center: Three-stage Clos network, with pairs of input and output vertices in consecutive stages connected by edges rather than being identified as single vertices. Right: Nine-stage network obtained by expanding each subunit of the center network into a three-stage network.}
\label{fig:clos}
\end{figure}

\begin{lemma}[Clos~\cite{Clo-BSTJ-53}]
\label{lem:clos}
Suppose that there exists a rearrangeable non-blocking network $G_c$ of capacity $c$ with $n$ vertices and $m$ edges.
Then there exists a rearrangeable non-blocking network of capacity $c^2$ with $3cn-2c^2$ vertices and $3cm$ edges.
\end{lemma}

\begin{proof}
Construct $3c$ copies of $G_c$, identified as $c$ input subunits, $c$ internal subunits, and $c$ output subunits.
The input subunits have together $c^2$ input vertices, which will be the inputs of the whole network. Similarly, the output subunits have together $c^2$ output vertices, which will be the outputs of the whole network. Identify each output vertex of an input subunit with an input vertex of an internal subunit, in such a way that each pair of these subunits has exactly one identified vertex.
Similarly, identify each output vertex of an internal subunit with an input vertex of an output subunit, in such a way that each pair of these subunits has exactly one identified vertex.

An example of this network, for $c=4$ and $G_c=K_{4,4}$, can be seen in an expanded form as the middle network of \cref{fig:clos}. For greater legibility of the figure, instead of identifying pairs of vertices between subunits, these pairs have been connected by added edges. Contracting these edges would produce the network described above.

To produce vertex-disjoint paths connecting any system of disjoint pairs of inputs and outputs, consider these pairs as defining a multigraph connecting the input subunits to the output subunits of the overall network. This multigraph has maximum degree~$c$ (each input or output subunit participates in at most $c$ pairs), and we may apply a theorem of Dénes Kőnig according to which every bipartite multigraph with maximum degree~$c$ has an edge coloring using $c$ colors~\cite{KHo-MA-16}. These colors may be associated with the $c$ internal subunits, and used to designate which internal subunit each path should pass through. Once this designation is made, each subunit has its own system of disjoint pairs of inputs and outputs through which its paths should go, and the paths through each subunit can be completed using the assumption that it is rearrangeable non-blocking.
\end{proof}

\begin{corollary}
\label{cor:clos}
For any constant $\varepsilon>0$ and any integer $c\ge 1$, there exist rearrangeable non-blocking networks of capacity $c$ with $O(c)$ vertices and $O(c^{1+\varepsilon})$ edges.
\end{corollary}

\begin{proof}
We prove the result by induction on the integer $i=\lceil\log_2 1/\varepsilon\rceil$.
As a base case this is true for $\varepsilon=1$ (for which $i=0$) and for arbitrary $c$, using the complete bipartite graph as the network. For smaller values of $\varepsilon$, apply the induction hypothesis with the parameters $2\varepsilon$ and $\lceil\sqrt c\rceil$, to produce a rearrangeable non-blocking network $N$ of capacity $\lceil\sqrt c\rceil$ with $O(\sqrt c)$ vertices and $O(c^{1/2+\varepsilon})$ edges. Applying \cref{lem:clos} to $N$ produces a rearrangeable non-blocking network of capacity $\ge c$ with $O(c)$ vertices and $O(c^{1+\varepsilon})$ edges, as desired.
Deleting excess vertices to reduce the capacity to exactly $c$ completes the induction.
\end{proof}

\begin{theorem}
For any $m$ and $n$ with $n\le m\le2\tbinom{n}{2}$, there exists a directed graph on $m$ edges and $n$ vertices on which any deterministic or high-probability randomized non-adaptive relaxation algorithm for shortest paths must use $\Omega(mn/\log n)$ relaxation steps. When $m=\Omega(n^{1+\varepsilon})$ for some $\varepsilon>0$, the lower bound improves to $\Omega(mn)$.
\end{theorem}

\begin{proof}
We construct a graph according to the construction outlined above, in which we choose a capacity $c$, set up two disjoint sets $S$ and $T$ of $c$ vertices, connect $T$ to $S$ by a biregular bipartite digraph of some degree $d$, and connect $S$ to $T$ by a rearrangeable non-blocking network of capacity $c$. We allocate at least $m/2$ edges to the biregular graph, and the rest to the non-blocking network, giving $d\approx m/2c$. For the $\Omega(mn/\log n)$ bound, we use the non-blocking network of \cref{lem:expander}, with $c=\Theta(n/\log n)$. For the $\Omega(mn)$ bound, we use the non-blocking network of \cref{cor:clos}, with $c=\Theta(n)$. In both cases, we can choose the parameters of these networks to achieve these asymptotic bounds without exceeding the given numbers $n$ and $m$ of vertices and edges. We pad the resulting graph with additional vertices and edges in order to make the numbers of vertices and edges be exactly $n$ and $m$, and set the weights of these padding edges to be high enough that they do not interfere with the remaining construction.

Next, we choose a random distribution on weights for the resulting network so that, for every relaxation sequence $\sigma$, the expected reduced cost of $\sigma$, for weights from this distribution, matches the lower bound in the statement of the lemma. For deterministic non-adaptive relaxation algorithms, this will give the desired lower bound directly, via the simple fact that the worst case of any distribution is always at least its expectation. For randomized algorithms, the lower bound will follow using \cref{cor:red} and \cref{cor:yao} to convert the lower bound on expected reduced cost into a high-probability lower bound.

As in \cref{thm:complete-random}, the random distribution on weights that we use is determined from a random distribution on paths from the source, such that the shortest path tree for the weighted graph will contain the chosen path. We can accomplish this by setting the lengths of the path edges to zero and all other edge lengths to one. Unlike in \cref{thm:complete-random}, these paths will not necessarily include all vertices in the graph and the shortest path tree may contain other branches. To choose a random path, we simply choose a sequence of edges in the biregular graph, one at a time, in order along the path. In each step, we choose uniformly at random among the subset of edges in the biregular graph that are disjoint from already-chosen edges. Because of the biregularity of the biregular part of our graph, each chosen edge is incident to at most $2(d-1)$ other edges, and eliminates these other edges from being chosen later. At least $c/2$ choices are possible before there are no more disjoint edges, and throughout the first $c/4$ choices there will remain at least $m/4$ edges to choose from, disjoint from all previous edges. The sequence ends when there are no more such edges to choose. Once we have chosen this sequence of edges from the biregular graph, we construct a set of vertex-disjoint paths in the rearrangeable nonblocking network that connects them in sequence into a single path.

For any given relaxation sequence $\sigma$, as in the proof of \cref{thm:complete-random}, let $\tau$ be the subsequence of edges in $\sigma$ that belong to the biregular part of the graph, and consider a modified relaxation algorithm that, after relaxing each edge in $\tau$, immediately relaxes all edges of the non-blocking network. Define the reduced cost for $\tau$ to be the number of relaxation steps made from $\tau$ before all distances are correct, not counting the relaxation steps in the non-blocking network. Clearly, this is at most equal to the reduced cost for $\sigma$, because $\sigma$ might fail to relax a path in the non-blocking network when $\tau$ succeeds, causing the computation of shortest path distances using $\sigma$ to fall behind that for $\tau$. Define $t_i$ to be the step in the relaxation sequence for $\tau$ that relaxes the $i$th chosen edge from the biregular graph, making the distance to its ending vertex correct. Then the expectation of $t_i-t_{i-1}$ (conditioned on the choice of the first $i-1$ edges  is at least the average, over all edges that were available to be chosen as the $i$th edge, of the number of steps along $\tau$ from $t_{i-1}$ to the next occurrence of that edge. This expectation is minimized when the edges occurring immediately following position $t_{i-1}$ in $\tau$ are exactly the next available edges, and is equal to half the number of available edges; for other possibilities for $\tau$, the expectation can only be even larger. The expected reduced cost for~$\tau$ equals the sum of these differences $t_i-t_{i-1}$. Since there are $\Omega(c)$ steps in which the number of available edges is $\Omega(m)$, the expected reduced cost for $\tau$ is $\Omega(cm)$. The expected reduced cost for~$\sigma$ can only be larger, and plugging in the value of $c$ (coming from our choice of which type of non-blocking network to use) gives the result.
\end{proof}

\section{Conclusions and Open Problems}

We have shown that, for a wide range of choices for $m$ and $n$, the Bellman–Ford algorithm is asymptotically optimal among non-adaptive relaxation algorithms. Adaptive versions of the Bellman–Ford algorithm are faster, but only by constant factors. Is it possible to prove that, among adaptive relaxation algorithms, Bellman–Ford is optimal? Doing so would require a careful specification of what information about the results of relaxation steps can be used in choosing how to adapt the relaxation sequence.

The constant factors of $\tfrac16$ and $\tfrac{1}{12}$ in our deterministic and randomized lower bounds for complete graphs are far from the constant factors of $\tfrac12$ and $\tfrac13$ in the corresponding upper bounds. Can these gaps be tightened? Is it possible to make them tight enough to distinguish deterministic and randomized complexity? Alternatively, is it possible to improve the deterministic methods to match the known randomized upper bound?

For sparse graphs ($m=O(n)$), our lower bound falls short of the Bellman–Ford upper bound by a logarithmic factor. Can the lower bound in this range be improved, or can the Bellman–Ford algorithm for sparse graphs be improved?

In this work, we considered the worst-case number of relaxation steps used by non-adaptive relaxation algorithms for the parameters $m$ and $n$. But it is also natural to look at this complexity for individual graphs, with unknown weights. For any given graph, there is some relaxation sequence that is guaranteed to find shortest path distances for all weightings of that graph, with as few relaxation steps as possible. An algorithm of Haddad and Schäffer~\cite{HadSch-SIDMA-88} can find such a sequence for the special case of graphs for which it is as short as possible, one relaxation per edge. What is the complexity of finding or approximating it more generally?
 
\section*{Acknowledgements}
This research was supported in part by NSF grant CCF-2212129.

\date{ }

\bibliographystyle{plainurl}
\bibliography{bellford}

\begin{thebibliography}{10}

\bibitem{AloCap-FOCS-07}
Noga Alon and Michael~R. Capalbo.
\newblock {Finding disjoint paths in expanders deterministically and online}.
\newblock In {\em 48th Annual IEEE Symposium on Foundations of Computer Science
  (FOCS 2007), October 20{--}23, 2007, Providence, RI, USA, Proceedings}, pages
  518{--}524. IEEE Computer Society, 2007.
\newblock \href {https://doi.org/10.1109/FOCS.2007.19}
  {\path{doi:10.1109/FOCS.2007.19}}.

\bibitem{BanEpp-ANALCO-12}
Michael~J. Bannister and David Eppstein.
\newblock {Randomized speedup of the Bellman{--}Ford algorithm}.
\newblock In Conrado Mart{\'\i}nez and Hsien{-}Kuei Hwang, editors, {\em
  Proceedings of the 9th Meeting on Analytic Algorithmics and Combinatorics,
  ANALCO 2012, Kyoto, Japan, January 16, 2012}, pages 41{--}47. SIAM, 2012.
\newblock \href {https://doi.org/10.1137/1.9781611973020.6}
  {\path{doi:10.1137/1.9781611973020.6}}.

\bibitem{Bei-SMTA-62}
B.~Beizer.
\newblock {The analysis and synthesis of signal switching networks}.
\newblock In {\em Proceedings of the Symposium on Mathematical Theory of
  Automata, New York, April 1962}, pages 563{--}572. Polytechnic Press of the
  Polytechnic Institute of Brooklyn, 1962.

\bibitem{Bel-QAM-58}
Richard Bellman.
\newblock {On a routing problem}.
\newblock {\em Quarterly of Applied Mathematics}, 16:87{--}90, 1958.
\newblock \href {https://doi.org/10.1090/qam/102435}
  {\path{doi:10.1090/qam/102435}}.

\bibitem{BerNanWul-FOCS-22}
Aaron Bernstein, Danupon Nanongkai, and Christian Wulff-Nilsen.
\newblock {Negative-weight single-source shortest paths in near-linear time}.
\newblock In {\em 63rd IEEE Annual Symposium on Foundations of Computer
  Science, FOCS 2022, Denver, CO, USA, October 31 {--} November 3, 2022}, pages
  600{--}611. IEEE, 2022.
\newblock \href {https://doi.org/10.1109/FOCS54457.2022.00063}
  {\path{doi:10.1109/FOCS54457.2022.00063}}.

\bibitem{CheAns-ICL-04}
Gang Cheng and Nirwan Ansari.
\newblock {Finding all hops shortest paths}.
\newblock {\em IEEE Communications Letters}, 8(2):122{--}124, 2004.
\newblock \href {https://doi.org/10.1109/LCOMM.2004.823365}
  {\path{doi:10.1109/LCOMM.2004.823365}}.

\bibitem{Clo-BSTJ-53}
Charles Clos.
\newblock {A study of non-blocking switching networks}.
\newblock {\em Bell System Technical Journal}, 32(2):406{--}424, March 1953.
\newblock \href {https://doi.org/10.1002/j.1538-7305.1953.tb01433.x}
  {\path{doi:10.1002/j.1538-7305.1953.tb01433.x}}.

\bibitem{DeoPan-Nw-84}
Narsingh Deo and Chi~Yin Pang.
\newblock {Shortest-path algorithms: taxonomy and annotation}.
\newblock {\em Networks}, 14(2):275{--}323, 1984.
\newblock \href {https://doi.org/10.1002/net.3230140208}
  {\path{doi:10.1002/net.3230140208}}.

\bibitem{Dia-CACM-69}
Robert~B. Dial.
\newblock {Algorithm 360: Shortest-path forest with topological ordering [H]}.
\newblock {\em Communications of the ACM}, 12(11):632{--}633, 1969.
\newblock \href {https://doi.org/10.1145/363269.363610}
  {\path{doi:10.1145/363269.363610}}.

\bibitem{ForFul-62}
L.~R. Ford, Jr. and D.~R. Fulkerson.
\newblock {A shortest chain algorithm}.
\newblock In {\em Flows in Networks}, pages 130{--}134. Princeton University
  Press, 1962.

\bibitem{GueOrd-IATW-02}
Roch Gu{\'e}rin and Ariel Orda.
\newblock {Computing shortest paths for any number of hops}.
\newblock {\em IEEE/ACM Transactions on Networkin}, 10(5):613{--}620, 2002.
\newblock \href {https://doi.org/10.1109/TNET.2002.803917}
  {\path{doi:10.1109/TNET.2002.803917}}.

\bibitem{HadSch-SIDMA-88}
Ramsey~W. Haddad and Alejandro~A. Sch{\"a}ffer.
\newblock {Recognizing Bellman{--}Ford-orderable graphs}.
\newblock {\em SIAM Journal on Discrete Mathematics}, 1(4):447{--}471, 1988.
\newblock \href {https://doi.org/10.1137/0401045} {\path{doi:10.1137/0401045}}.

\bibitem{RFC-1058}
Charles Hedrick.
\newblock {\em {Routing Information Protocol}}.
\newblock Request for Comments, RFC 1058. Network Working Group, June 1988.
\newblock URL: \url{https://www.rfc-editor.org/rfc/rfc1058}.

\bibitem{JukSch-TCS-16}
Stasys Jukna and Georg Schnitger.
\newblock {On the optimality of Bellman-Ford-Moore shortest path algorithm}.
\newblock {\em Theoretical Computer Science}, 628:101{--}109, 2016.
\newblock \href {https://doi.org/10.1016/j.tcs.2016.03.014}
  {\path{doi:10.1016/j.tcs.2016.03.014}}.

\bibitem{Pol-22}
Tomasz Kociumaka and Adam Polak.
\newblock {Bellman{--}Ford is optimal for shortest hop-bounded paths}.
\newblock Electronic preprint arxiv:2211.07325, February 14 2023.

\bibitem{KHo-MA-16}
D{\'e}nes K{\H{o}}nig.
\newblock {{\"U}ber Graphen und ihre Anwendung auf Determinantentheorie und
  Mengenlehre.}
\newblock {\em Mathematische Annalen}, 77:453{--}465, 1916.
\newblock \href {https://doi.org/10.1007/BF01456961}
  {\path{doi:10.1007/BF01456961}}.

\bibitem{MeyNegWei-TAPAS-11}
Ulrich Meyer, Andrei Negoescu, and Volker Weichert.
\newblock {New bounds for old algorithms: on the average-case behavior of
  classic single-source shortest-paths approaches}.
\newblock In Alberto Marchetti-Spaccamela and Michael Segal, editors, {\em
  Theory and Practice of Algorithms in (Computer) Systems {--} First
  International ICST Conference, TAPAS 2011, Rome, Italy, April 18{--}20, 2011.
  Proceedings}, volume 6595 of {\em Lecture Notes in Computer Science}, pages
  217{--}228. Springer, 2011.
\newblock \href {https://doi.org/10.1007/978-3-642-19754-3_22}
  {\path{doi:10.1007/978-3-642-19754-3_22}}.

\bibitem{Moo-ISST-57}
Edward~F. Moore.
\newblock {The shortest path through a maze}.
\newblock In {\em Proc. Internat. Sympos. Switching Theory 1957, Part II},
  pages 285{--}292. Harvard Univ. Press, Cambridge, Mass., 1959.

\bibitem{Pip-JCSS-78}
Nicholas Pippenger.
\newblock {On rearrangeable and non-blocking switching networks}.
\newblock {\em Journal of Computer and System Sciences}, 17(2):145{--}162,
  1978.
\newblock \href {https://doi.org/10.1016/0022-0000(78)90001-6}
  {\path{doi:10.1016/0022-0000(78)90001-6}}.

\bibitem{Yao-FOCS-77}
Andrew Chi-Chih Yao.
\newblock {Probabilistic computations: toward a unified measure of complexity}.
\newblock In {\em 18th Annual Symposium on Foundations of Computer Science,
  Providence, Rhode Island, USA, 31 October {--} 1 November 1977}, pages
  222{--}227. IEEE Computer Society, 1977.
\newblock \href {https://doi.org/10.1109/SFCS.1977.24}
  {\path{doi:10.1109/SFCS.1977.24}}.

\bibitem{Yen-75}
Jin~Y. Yen.
\newblock {\em {Shortest Path Network Problems}}, volume~18 of {\em
  Mathematical Systems in Economics}.
\newblock Verlag Anton Hain, Meisenheim am Glan, 1975.

\end{thebibliography}

\end{document}